\documentclass[aps,prb,twocolumn,superscriptaddress]{revtex4-2}
\usepackage{amsfonts, amssymb, amsmath,amsthm,graphicx}
\usepackage{subfigure}
\usepackage{multirow}
\usepackage{dcolumn}
\usepackage{bm}
\usepackage{color}
\usepackage[usenames,dvipsnames]{xcolor}

\renewcommand{\bf}[1]{\textnormal{\textbf{#1}}}
\newcommand{\BZ}{\textnormal{\text{BZ}}}
\newcommand{\tr}{\textnormal{\text{Tr}}}
\newcommand{\vol}{\textnormal{\text{vol}}}
\newcommand{\Gr}{\textnormal{\text{Gr}}}

\newtheorem{proposition}{Proposition}
\newtheorem{theorem}{Theorem}
\newtheorem{corollary}{Corollary}
\newtheorem{remark}{Remark}

\newcommand{\ket}[1]{| #1 \rangle}
\newcommand{\bra}[1]{\langle #1|}

\usepackage[colorlinks=true,allcolors=blue]{hyperref}

\begin{document}

\title{K\"{a}hler geometry and Chern insulators: Relations between topology and the quantum metric}

\author{Bruno Mera}
\affiliation{Instituto de Telecomunica\c{c}\~oes, 1049-001 Lisboa, Portugal}
\affiliation{Departmento de F\'{i}sica, Instituto Superior T\'ecnico, Universidade de Lisboa, Av. Rovisco Pais, 1049-001 Lisboa, Portugal}
\affiliation{Departmento de Matem\'{a}tica, Instituto Superior T\'ecnico, Universidade de Lisboa, Av. Rovisco Pais, 1049-001 Lisboa, Portugal}
\author{Tomoki Ozawa}
\affiliation{Advanced Institute for Materials Research (WPI-AIMR), Tohoku University, Sendai 980-8577, Japan}

\date{\today}

\newcommand{\tom}[1]{{\color{red} #1}}
\definecolor{bbblue}{rgb}{0.,0.24,0.51}
\newcommand{\blue}{\color{bbblue}}

\begin{abstract}
We study Chern insulators from the point of view of K\"{a}hler geometry, i.e. the geometry of smooth manifolds equipped with a compatible triple consisting of a symplectic form, an integrable almost complex structure and a Riemannian metric. The Fermi projector, i.e. the projector onto the occupied bands, provides a map to a K\"{a}hler manifold. The quantum metric and Berry curvature of the occupied bands are then related to the Riemannian metric and symplectic form, respectively, on the target space of quantum states. We find that the minimal volume of a parameter space with respect to the quantum metric is $\pi |\mathcal{C}|$, where $\mathcal{C}$ is the first Chern number. We determine the conditions under which the minimal volume is achieved both for the Brillouin zone and the twist-angle space. The minimal volume of the Brillouin zone, provided the quantum metric is everywhere non-degenerate, is achieved when the latter is endowed with the structure of a K\"{a}hler manifold inherited from the one of the space of quantum states. If the quantum volume of the twist-angle torus is minimal, then both parameter spaces have the structure of a K\"{a}hler manifold inherited from the space of quantum states. These conditions turn out to be related to the stability of fractional Chern insulators. For two-band systems, the volume of the Brillouin zone is naturally minimal provided the Berry curvature is everywhere non-negative or nonpositive, and we additionally show how the latter, which in this case is proportional to the quantum volume form, necessarily has zeros due to topological constraints.
\end{abstract}

\maketitle
\section{Introduction}
The notion of topological phases has drastically changed our understanding of gapped phases of matter. There is much to learn about a band insulator beyond the assertion that it has a gap separating the valence bands, which form the occupied bands at zero temperature, from the conduction bands. In the particular case of two spatial dimensions, in the absence of other symmetries, the occupied bands may have a non-trivial topological ``twist'' determining what is called a Chern insulator. This topological twist is not merely a mathematical observation and it is manifested in the topological response of the system as known in integer and fractional quantum Hall effects, in which the transverse Hall conductivity is proportional to the first Chern number, a topological invariant of the vector bundle of occupied states over the two-dimensional Brillouin zone.
Finer properties of insulating states refer not just to the topology, but also to their geometry~\cite{res:11,res:18}. Of particular interest is the momentum-space quantum metric, related to the overlap of Bloch states at nearby quasi-momenta~\cite{res:11,res:18,mat:ryu:10}. As it turns out, the integral of this metric over the Brillouin zone~\cite{oza:gol:19} is associated to the spread functional and the localization tensor of the material~\cite{mar:van:97, res:98, res:sor:99, sou:wil:mar:00, mat:ryu:10}. There are numerous proposals in the literature for extracting the quantum metric~\cite{kol:gri:pol:13, neu:cha:mud:13, oza:18, oza:gol:18, ble:sol:mal:18, ble:mal:gao:sol:18, kle:ras:cue:bel:18}.
A proposal~\cite{oza:gol:19}, making use of the fluctuation-dissipation theorem, shows that the localization tensor can also be measured through spectroscopy in synthetic quantum matter, such as ultracold atomic gases or trapped ions.
The momentum-space quantum metric has recently been observed in ultracold gases~\cite{Asteria:2019} and in a microcavity hosting exciton-polariton modes~\cite{Gianfrate:2020}.

The connection of microscopic geometric quantities, such as the Brillouin zone Berry curvature and quantum metric, to macroscopic geometric quantities associated with ground state properties of the band insulator is given, at zero temperature, by integrating the former quantities over the Brillouin zone. Equivalently, one can consider the system in finite size and take periodic boundary conditions twisted by phases. As one varies these phases, one obtains a family of many-particle ground states over the so-called twist-angle space and from which we can extract the macroscopic geometric quantities mentioned above. In two spatial dimensions, the Berry curvature in twist-angle space yields the famous Thouless--Kohomoto--Nightingale--den Nijs result for the Hall conductivity~\cite{wit:16,kud:wat:kar:hat:19,mer:20:1} and the quantum metric is related to the localization tensor~\cite{oza:gol:19,mer:20:1}. The anisotropy of the localization tensor is measured by a modular parameter $\tau$ in the upper half-plane which describes a complex structure in twist-angle space~\cite{mer:20:1}.

The study of the geometry of band insulators can also be used to understand if the material can host stable fractional topological phases~\cite{roy:14, jac:mol:roy:15}. In particular, for a Chern band to have an algebra of projected density operators which is isomorphic to the $W_{\infty}$ algebra found by Girvin, MacDonald, and Platzman for the fractional quantum Hall effect~\cite{gir:mac:pla:86}, certain geometric constraints are naturally found which enforce a compatibility relation between the momentum-space quantum metric and the Berry curvature of the band~\cite{roy:14,jac:mol:roy:15}. This compatibility relation between the quantum metric and the Berry curvature is, as noted in Ref.~\cite{jac:mol:roy:15}, the same compatibility required, in an oriented surface, between the Riemannian metric and a symplectic form for these to endow the surface with a K\"{a}hler structure.

Motivated by the above, we provide in this paper a detailed study of the geometry of Chern insulators from the perspective of K\"{a}hler geometry. In particular, we find that the minimal quantum volume, as determined by the Riemannian volume form induced by the quantum metric, is given by $\pi |\mathcal{C}|$, where $\mathcal{C}$ is the first Chern number of the occupied band. Furthermore, we determine the conditions of achieving the minimal quantum volume both in the Brillouin zone and in the twist-angle space. The minimal volume of the Brillouin zone, provided the quantum metric is everywhere non-degenerate, is achieved when the latter is endowed with the structure of a K\"{a}hler manifold induced from the one of the space of quantum states. If the quantum volume of the twist-angle torus is minimal, then both tori have the structure of a K\"{a}hler manifold inherited from the space of quantum states. These conditions are related to the geometric stability conditions for fractional Chern insulators presented in Refs.~\cite{roy:14,jac:mol:roy:15}, measuring the deviation from the lowest Landau-level physics. For two-band systems, the quantum volumes are naturally minimal provided the Berry curvature is everywhere non-negative or nonpositive, and we additionally show how the latter, which in this case is proportional to the quantum volume form, necessarily has zeros due to topological constraints. The results of this paper are accompanied, complemented, and supported by those of Ref.~\cite{oz:mer:21:published}, in which physical implications of the mathematical structures we find in this paper are studied and verified with explicit models.

This paper is organized as follows. In Sec.~\ref{sec: the geometry of chern insulators}, we explore the geometry of Chern insulators from the point of view of K\"{a}hler geometry. In Sec.~\ref{sec: main results}, we present our main results, namely Theorems~\ref{th:kahlerband},~\ref{th:tomoki} and~\ref{th:bruno}, together with some corollaries and remarks. In Sec.~\ref{sec: physical interpretation of the Kahler condition}, we provide a physical interpretation for the results of Theorem~1 in terms of the many-particle ground state of the band insulator. Finally in Sec.~\ref{sec: conclusions} we draw the conclusions. In Appendix~\ref{Appendix:A}, the reader can find detailed mathematical derivations of the results presented in Sec.~\ref{sec: main results}.

\section{The geometry of Chern insulators}
\label{sec: the geometry of chern insulators}
In a tight-binding description of an insulator in two spatial dimensions, in the presence of translation invariance, the Hamiltonian is described by an $n\times n$ matrix $H(\bf{k})$ which is smooth as a function of quasimomentum $\bf{k}$ in the Brillouin zone $\BZ^2$, topologically, a two-torus. The number $n$ specifies the internal degrees of freedom. The presence of a gap below the Fermi level $E_F$ allows one to define the Fermi projector $P(\bf{k})=\Theta(E_F-H(\bf{k}))$, where $\Theta(\cdot)$ is the Heaviside step function, which describes the ground state of the insulator and, up to adiabatic deformation, it determines its topological properties. In fact, the smoothness and the gap condition imply that we have a rank $r$ Hermitian vector bundle $E\to\BZ^2$, where $r$ is the number of occupied bands and is the constant rank of the Fermi projector, whose fiber over $\bf{k}\in\BZ^2$ is simply $\text{Im}\left[P(\bf{k})\right]\subset \mathbb{C}^n$. The smooth family of orthogonal projectors $\{P(\bf{k})\}_{\bf{k}\in\BZ^2}$ can be seen as a classifying map $P:\BZ^2\to \Gr_{k}(\mathbb{C}^n)$ for the vector bundle $E\to\BZ^2$, where we interpret the Grassmannian manifold of $r$-dimensional subspaces of $\mathbb{C}^n$, denoted $\Gr_r(\mathbb{C}^n)$, as the set of orthogonal projectors of rank $r$. The Chern insulator is determined by the reduced complex \emph{K}-theory class $[E]\in\widetilde{K}^0(\BZ^2)$, which is determined by the \emph{idempotent} $P$~\cite{kit:09,par:08,mer:20:2}.

The space $\Gr_{r}(\mathbb{C}^n)$ has the natural structure of a K\"{a}hler manifold (see, for instance, Refs.~\cite{can:01,nakahara} for more details), i.e., it is a symplectic manifold together with a compatible complex structure $J$ giving it the structure of a complex manifold. By a symplectic manifold structure we mean that it is a smooth manifold together with a non-degenerate closed two-form $\omega$ -- the symplectic form. The almost complex structure $J$ is a linear map on the tangent spaces that squares to minus the identity. Morally, it acts as multiplication by $\sqrt{-1}$. The fact that $\Gr_r(\mathbb{C}^n)$ is a complex manifold means that we can find local complex coordinates $z^j$ and, with respect to these,
\begin{align}
J\left(\frac{\partial}{\partial z^j}\right)=i\frac{\partial}{\partial z^j} \text{ and } J\left(\frac{\partial}{\partial \overline{z}^j}\right)=-i\frac{\partial}{\partial \overline{z}^j}.
\end{align}
Finally, the fact that $\omega$ and $J$ are compatible means that $g=\omega(\cdot,J\cdot)$ is a Riemannian metric. It is said in this case that the complex structure $J$ is $\omega$-compatible. In a K\"{a}hler manifold, the symplectic form is also known as the \emph{K\"{a}hler form}. In the particular case of the Grassmannian, the compatible triple of structures stem from the usual Fubini-Study K\"{a}hler structure on the complex projective space $\mathbb{C}P^{m}$, with $m=\binom{n}{r}-1$, through the Pl\"{u}cker embedding $\iota: \Gr_{r}(\mathbb{C}^n)\hookrightarrow \mathbb{P}(\Lambda^r\mathbb{C}^n)\cong \mathbb{C}P^{m}$ (see Ref.~\cite{huy:05}). The latter map sends $\text{Im}(P)$, or equivalently, $P$ to the line spanned by a Slater determinant made of a linear basis for $\text{Im}(P)$. For this reason we write these structures as $\omega_{FS}, J_{FS}$, and $g_{FS}$.

One can write these three structures in a familiar form by using orthogonal projectors. To do this, it is useful to recall that, since a projector satisfies $P^2=P$, if $Q=I-P$ is the orthogonal complement projector, where $I$ is the identity matrix, then we have the relation $PdPP=QdPQ=0$ and hence the differential $dP$ satisfies $dP=QdPP+PdPQ$. Furthermore, since $P^\dagger=P$ for orthogonal projectors, we see that $dP$ is completely  determined by the triangular matrix $QdPP$ [in an orthonormal basis where $P=\text{diag}(I_{r},0_{n-r})$ it is represented by a lower triangular matrix whose entries are one-forms].

The symplectic and Riemannian structures are given by
\begin{align}
\omega_{FS} &=-\frac{i}{2}\tr\left(\left(QdPP\right)^\dagger\wedge QdPP\right)=-\frac{i}{2}\tr\left(PdP\wedge dP\right) \nonumber ,\\
g_{FS} &=\tr\left(\left(QdPP\right)^{\dagger}QdPP\right)=\tr\left(PdPdP\right),
\end{align}
and the complex structure is defined by
\begin{align}
\left(QdPP\right)\circ J_{FS}\!=\!iQdPP \!\text{ and }\! \left(PdPQ\right)\circ J_{FS}\!=\!-iPdPQ,
\end{align}
where the second equation is essentially the Hermitian conjugate of the first since $J_{FS}$ is an $\mathbb{R}$-linear operator and, therefore, it is completely specified by the first equation. It is not hard to see that, indeed, $\omega_{FS}(\cdot, J_{FS}\cdot)=g_{FS}$. A remarkable property of this structure is that it is invariant under the action of $\text{U}(n)$ by conjugation $P\mapsto UPU^\dagger$, for $U\in\text{U}(n)$, so all the structures are completely determined by knowing them at one particular projector. Before moving on, it is important to note that $\omega_{FS}$ is related to the curvature two-form of a connection on a line bundle over $\Gr_{r}(\mathbb{C}^n)$. Namely, over $\Gr_{r}(\mathbb{C}^n)$ we can consider the vector bundle whose fiber over an orthogonal projector $P$ is simply $\text{Im}(P)$ --- the tautological bundle $E_0\to \Gr_{r}(\mathbb{C}^n)$. This rank $r$-vector bundle is naturally a subbundle of the trivial bundle $\theta^n=\Gr_{r}(\mathbb{C}^n)\times\mathbb{C}^n$. The Berry connection is then simply given by projection of the exterior derivative. The curvature of this connection is determined by the matrix of two-forms
\begin{align}
\Theta=iPdP\wedge dPP.
\end{align}
We can then consider the top exterior power of this bundle, $\Lambda^r E_0\to\Gr_r(\mathbb{C}^n)$, which is a line bundle whose fiber at $P$ is spanned by a Slater determinant
\begin{align}
v_1\wedge \dots \wedge v_{r}, \text{ with } \text{span}_{\mathbb{C}}\{v_{i}\}_{i=1}^r=\text{Im}(P).
\label{eq: Slater det}
\end{align}
The connection induced on this line bundle from the Berry connection on the tautological vector bundle has curvature two-form
\begin{align}
\Omega_0=\tr\left(\Theta\right)=i\tr\left(PdP\wedge dPP\right)=i\tr\left(PdP\wedge dP\right).
\end{align}
The relation between $\omega_{FS}$ and $\Omega_0$ is then clear: $\omega_{FS}=-(1/2)\Omega_0$. In the physics literature, the quantity $\chi=g_{FS}+i\omega_{FS}=g_{FS}-i(\Omega_0/2)$ is usually known as the quantum geometric tensor.

Geometric quantities related to $g_{FS}$ and $\omega_{FS}$ appear naturally in the context of insulators because we have a map $P:\BZ^2\to \Gr_{r}(\mathbb{C}^n)$, which provides the means to pull back the rich geometry in the target to the Brillouin zone. The pullback by $P$ of the tautological vector bundle, i.e., the vector bundle whose fiber at $\bf{k}\in \BZ^2$ is $\text{Im}\left[P(\bf{k})\right]$, is simply the occupied band vector bundle $P^*E_0=E\to\BZ^2$. The pullback of the $r$th exterior power of the tautological vector bundle is a line bundle whose fiber at $\bf{k}$ is spanned by a Slater determinant of a basis for $\text{Im}\left[P(\bf{k})\right]$. This bundle specifies the ground state completely as, physically, the ground state is obtained precisely by occupying the $r$ bands below the Fermi level $E_F$.

The Berry curvature for $\Lambda^rE=P^*\Lambda^rE_0$ is
\begin{align}
\Omega=P^*\Omega_0=i\sum_{i,j=1}^2\tr\left(P(\bf{k})\frac{\partial P}{\partial k_i}(\bf{k})\frac{\partial P}{\partial k_j}(\bf{k})\right)dk_i\wedge dk_j.
\label{eq: momentum-space Berry curvature}
\end{align}
Meanwhile, the pullback of the metric $g_{FS}$ by $P$, is the so-called quantum metric
\begin{align}
g=P^*g_{FS}=\sum_{i,j=1}^2\tr\left(P(\bf{k})\frac{\partial P}{\partial k_i}(\bf{k})\frac{\partial P}{\partial k_j}(\bf{k})\right)dk_idk_j.
\label{eq: momentum-space quantum metric}
\end{align}

Following the discussion in Ref.~\cite{mer:20:1}, we can consider threading a flux through a finite system of size $N\times N$  with periodic boundary conditions, i.e., in a torus. This means that the fermions acquire phases $e^{i\theta_i}$, $i=1,2$, as they are adiabatically moved around the fundamental cycles of position space. The angles $\bm{\theta}=(\theta_1,\theta_2)$ are referred to as the twist-angles and they span the twist-angle torus $T^2_{\theta}$. In momentum space, these angles can be accounted for by sampling the matrix $H(\bf{k})$ at points $\bf{k}=(2\pi/N)\bf{m}+\bm{\theta}/N$, with $\bf{m}\in\{0,\dots,N-1\}^2$. The ground state of the system, as we vary $\bm{\theta}$, produces a smooth Hermitian line bundle $\mathcal{L}\to T^2_{\theta}$, whose Berry curvature has the form
\begin{align}
\widetilde{\Omega}=\frac{1}{N^2}\sum_{\bf{m}}\Omega_{12}\left(\frac{2\pi \bf{m}}{N}+\frac{\bm{\theta}}{N}\right)d\theta_1\wedge d\theta_2,
\end{align}
where $\Omega_{12}=\Omega(\partial/\partial k_1,\partial/\partial k_2)$ is the single component of the Berry curvature in Eq.~\eqref{eq: momentum-space Berry curvature}. In the thermodynamic limit, $N\to\infty$, we get
\begin{align}
\widetilde{\Omega}\to \left(\int_{\BZ^2}\frac{d^2k}{(2\pi)^2}\Omega_{12}(\bf{k})\right)d\theta_1\wedge d\theta_2=\frac{\mathcal{C}}{2\pi }d\theta_1\wedge d\theta_2,
\end{align}
where $\mathcal{C}=\int_{\BZ^2}\Omega/(2\pi)$ is the first Chern number of the occupied Bloch bundle $E\to\BZ^2$. The relation to linear response theory is then
\begin{align}
e^2\widetilde{\Omega}=\sigma_{\text{Hall}}d\theta_1\wedge d\theta_2,
\end{align}
where $\sigma_{\text{Hall}}=e^2 \mathcal{C}/(2\pi)$ is the Hall conductivity of the insulator and $e$ is the charge of the fermions. We find that the Berry curvature is, in the thermodynamic limit $N\to\infty$, flat. Observe additionally that
\begin{align}
\mathcal{C}= \int_{\BZ^2}\frac{ \Omega}{2\pi}=\int_{T^2_{\theta}}\frac{\widetilde{\Omega}}{2\pi},
\end{align}
and, as a consequence, the line bundle $\mathcal{L}\to T^2_{\theta}$ is isomorphic, in the smooth category, to the line bundle $\Lambda^rE\to \BZ^2$, where we identify  $\BZ^2\cong T^2_{\theta}$, through $\bf{k}=(k_1,k_2)\mapsto (\theta^1(\bf{k}),\theta^2(\bf{k}))=(k^1,k^2)$.
In a similar fashion, we can determine the quantum metric in twist-angle space, yielding,
\begin{align}
\widetilde{g}=\frac{1}{N^2}\sum_{i,j=1}^{2}\sum_{\bf{m}}g_{ij}\left(\frac{2\pi \bf{m}}{N}+\frac{\bm{\theta}}{N}\right)d\theta_i d\theta_j,
\end{align}
where $g_{ij}=g(\partial/\partial k_i,\partial/\partial k_j)$ are the components of the quantum metric in Eq.~\eqref{eq: momentum-space quantum metric}. In the same way that the Berry curvature is flat, also the quantum metric is flat in the thermodynamic limit $N\to\infty$
\begin{align}
\widetilde{g}=\sum_{i,j=1}^2\left(\int_{\BZ^2}\frac{d^2k}{(2\pi)^2}g_{ij}(\bf{k})\right)d\theta_id\theta_j.
\label{eq: twist-angle quantum metric}
\end{align}
This metric also appears in the properties of the insulator in question, namely, it is related to the so-called localization tensor. The fact that we have a flat metric tensor, allowed us in Ref.~\cite{mer:20:1}, to extract a modular parameter $\tau\in \mathcal{H}$, which determines a complex structure $\widetilde{j}$ on $T^2_{\theta}$. This complex structure is defined by a rotation of 90 degrees in the standard orientation of the torus with respect to the quantum metric $\widetilde{g}$. This complex structure turned out to be related to the anisotropy of the localization tensor.

In two dimensions having a metric and an orientation specifies a complex structure. Indeed, provided $g$ is non-degenerate, we get, by a rotation of 90 degrees in the standard orientation of $\BZ^2$ with respect to $g$, an almost complex structure over it which we call $j$. We can get an explicit expression for $j$ as follows. At a given point $\bf{k}$ of the Brillouin zone we can always choose an orthonormal frame $e_1,e_2$, consistent with the orientation (i.e., the determinant of the matrix relating $e_1,e_2$ to $\partial/\partial k_1, \partial/\partial k_2$ is positive), for the quantum metric $g$, which is a Riemannian metric, so that
\begin{align*}
g(e_i,e_j)=\delta_{ij},\ i,j=1,2.
\end{align*}
Then at that point $j(e_1)=e_2$ and $j(e_2)=-e_1$. Hence, it follows that
\begin{align}
g\left(j(e_i),e_{j}\right)=\varepsilon_{ij}, \ i,j=1,2,
\end{align}
where $\varepsilon_{ij}$ is the Levi-Civita symbol in two dimensions. Equivalently, $g(j\cdot ,\cdot)= e^1\wedge e^2$ is the volume form associated with $g$, where $e^1,e^2$ is the dual basis of $e_1,e_2$. We will now give an expression for $j$ in the original $(k_1,k_2)$ coordinates. Let us denote by $g_{ij}$ and $j^{i}_j$, $1\leq i,j\leq 2$, the components of $g$ and $j$ in the $(k_1,k_2)$ coordinates, respectively.  From $g(j\cdot ,\cdot)=e^{1}\wedge e^{2}=\sqrt{\det(g)}dk_1\wedge dk_2$, where $\det(g)=\det (g_{ij})$, we get,
\begin{align}
\sum_{k=1}^{2}g_{kj}j^{k}_{i}=\sqrt{\det(g)}\varepsilon_{ij}, \ i,j=1,2.
\end{align}
It is also useful to write the equivalent form
\begin{align}
g_{ij}=\sum_{k=1}^2\sqrt{\det(g)}\varepsilon_{ik}j^{k}_{j} ,\ i,j=1,2.
\label{eq: (g,j,dvol_g) is compatibile}
\end{align}
Either way, it follows that
\begin{align}
j^{i}_{j}=-\frac{1}{\sqrt{\det(g)}}\sum_{k=1}^2\varepsilon^{ik}g_{kj},\ i,j=1,2,
\end{align}
or, in matrix form,
\begin{align}
j=\frac{1}{\sqrt{\det(g)}}\left[\begin{array}{cc}-g_{12} & -g_{22}\\g_{11} & g_{12}\end{array}\right]. \label{eq:j}
\end{align}
By construction, $j$ does indeed satisfy the requirements of an almost complex structure, namely that $j^2=-I$, provided $g$ is non-degenerate. The expression presented above for $j$ is global, i.e., we can use it at any point over the Brillouin zone. However, being able to find local coordinates $x_1,x_2$ in a neighborhood of any point $\bf{k}\in\BZ^2$ such that $j$ assumes the canonical form $j(\partial/\partial x_1)=\partial/\partial x_2$ and $j(\partial/\partial x_2)=-\partial/\partial x_1$ is a much more subtle question, namely that of integrability of $j$, and it implies solving a partial differential equation. If such local coordinates exist, then $z=x_1+ix_2$ will be a local complex coordinate since $j(\partial/\partial z)=i\partial/\partial z$ and $j(\partial/\partial \overline{z})=-i\partial/\partial \overline{z}$. What we did above by choosing a orthonormal basis at a point (which is a linear combination of the natural tangent vectors $\partial/\partial k_1$ and $\partial/\partial k_2$) was to make $j$ look canonical at a specific point. Indeed we can always choose local coordinates $x_i(k_1,k_2)$, $i=1,2$, such that $\partial/\partial x_i$, $i=1,2$, become, at a specific point, an orthonormal frame for $g$ but away from that point that will, in general, not happen. Finding local coordinates such that this happens everywhere, up to a conformal factor, is equivalent to solving the Beltrami equation for isothermal coordinates. Hence, solving the Beltrami equation gives us local holomorphic coordinates in two dimensions. Being able to find local holomorphic coordinates for which $j$ assumes the canonical form of multiplication by $i$ is a very special feature of oriented surfaces. Indeed, in two dimensions any almost complex structure $j$ is integrable and it endows the surface with the structure of a complex manifold~\cite{che:67}. One is then naturally lead to ask under which conditions the geometry induced by $P$ on $\BZ^2$, given by the triple $(P^*\omega_{FS},j,P^{*}g_{FS})=(\omega,j,g)$, is compatible and gives rise to a K\"{a}hler structure. In this case, the map $P:\BZ^2\to \Gr_{r}(\mathbb{C}^n)$ is a K\"{a}hler map in the sense that the pullback of the triple $(\omega_{FS},J_{FS},g_{FS})$ is well-defined as a compatible triple $(\omega=P^*\omega_{FS},j=P^{*}J_{FS},g=P^*g_{FS})$ giving $\BZ^2$ the structure of a K\"{a}hler manifold. The meaning of $j=P^*J_{FS}$ needs further justification. In the case that $g=P^*g_{FS}$ is a Riemannian metric the map $P$ has to be an immersion, i.e., the differential $dP$ is full-rank at every point of the Brillouin zone. In that case $dP$ provides an isomorphism of the tangent spaces of the Brillouin zone to their images. The condition that $\omega (\cdot, j \cdot)$, with $\omega$ non-degenerate, then turns out to imply $dP \circ j = J_{FS} \circ dP$, i.e., the map is holomorphic --- this is what we mean by $j = P^* J_{FS}$. We will see that when $(\omega, j, g)$ form a compatible triple, the Chern number is necessarily negative. We point out that, under the condition of non-degeneracy of $g$, the Riemannian volume form $\sqrt{\det(g)}dk_1\wedge dk_2$ (notice that the orientation of $\BZ^2$ is explicitly used to define it), where $\det(g)$ is the determinant of the matrix representing the metric in periodic coordinates, is symplectic, i.e., it is non-degenerate and closed. As a consequence, if we consider the triple of structures $(\sqrt{\det(g)}dk_1\wedge dk_2,j,g)$, we  easily see that it is compatible, cf. Eq.~\eqref{eq: (g,j,dvol_g) is compatibile}, and it gives $\BZ^2$ a K\"{a}hler structure. However, this triple does not generally coincide with the triple $(\omega,j,g)$. We will see that only when the map is K\"{a}hler these two triples coincide. One can then ask what are the physical consequences of this condition. Furthermore, we could also ask whether $(\widetilde{\omega}=-\widetilde{\Omega}/2,\widetilde{j},\widetilde{g})$ gives $T^2_{\theta}$ the structure of a K\"{a}hler manifold and what are the physical consequences.

As we will see, there will also be cases where $(\omega, -j, g)$ form a compatible triple. Namely, the almost complex structure is defined with respect to the orientation opposite to the natural orientation specified by the ordered pair of coordinates $(k_1, k_2)$. This situation corresponds to the case where the Chern number is positive.

We note that the relation between holomorphicity of the wave function, related to the holomorphicity of the map to the projective space (see Sec.~\ref{sec: physical interpretation of the Kahler condition}),  and the Berry curvature being a K\"ahler form has been studied in the different context of the fractional quantum Hall effect in curved background geometry~\cite{ngu:can:gro:17}, where the relevant metric there is that of the spatial sample as opposed to the quantum metric.
\section{Main results}
\label{sec: main results}
In this section we answer the questions formulated in the end of previous section. We refer the reader to Appendix~\ref{Appendix:A} for the proofs.

We first note that the condition of $P:\BZ^2\to \Gr_r(\mathbb{C}^n)$ to be an immersion is equivalent to $g = P^* g_{FS}$ being non-degenerate [$\det (g) \neq 0$ everywhere], which we use frequently in the discussion below.

Next, we observe that for $P:\BZ^2\to \Gr_r(\mathbb{C}^n)$ to be K\"{a}hler it is enough for $P$ to be an immersion and to be holomorphic, i.e., to satisfy
\begin{align}
Q\frac{\partial P}{\partial \overline{z}}=0,
\end{align}
in local complex coordinates, see Proposition~\ref{prop: Kahler sufficient conditions} and its proof in Appendix~\ref{Appendix:A}. We then have the following result:

\newtheorem*{th:kahlerband}{Theorem \ref{th:kahlerband}}
\begin{th:kahlerband}
 Suppose $P:\BZ^2\to\Gr_{r}(\mathbb{C}^n)$ is an immersion, then the Cauchy-Schwarz like inequality is saturated
\begin{align*}
\sqrt{\det(g_{ij}(\bf{k}))}=\frac{\mp \Omega_{12}(\bf{k})}{2}, \text{ for all } \bf{k}\in \BZ^2,
\end{align*}
if and only if the map $P$ is K\"{a}hler with respect to the triple of structures $(\omega, \pm j, g)$.
\end{th:kahlerband}

We will use the notation $\vol_{g}(\BZ^2)$ and $\vol_{\widetilde{g}}(T^2_{\theta})$ to denote the integrals
\begin{align*}
\vol_{g}(\BZ^2)=\int_{\BZ^2}\sqrt{\det(g)}d^2k
\end{align*}
and
\begin{align*}
\vol_{\widetilde{g}}(T^2_{\theta})=\int_{T^2_{\theta}}\sqrt{\det(\widetilde{g})}d^2\theta,
\end{align*}
where $\det(g)$ and $\det(\widetilde{g})$ are the determinants of the matrices representing $g$ and $\widetilde{g}$ in the periodic coordinates $\bf{k}=(k_1,k_2)$ and $\bm{\theta}=(\theta_1,\theta_2)$, respectively. In the case where the quantum metrics are non-degenerate the quantities $\vol_{g}(\BZ^2)$ and $\vol_{\widetilde{g}}(T^2_{\theta})$ correspond to the volumes of $\BZ^2$ and $T^2_{\theta}$, respectively, as measured by the quantum metrics. For this reason, we will refer to them as the quantum volumes and the top-forms $\sqrt{\det(g)}dk_1\wedge dk_2$ and $\sqrt{\det(\widetilde{g})}d\theta_1\wedge d\theta_2$ as the quantum volume-forms of the Brillouin zone and twist-angle space, respectively. Furthermore we will use the diffeomorphism $\BZ^2\cong T^2_{\theta}$ given by  $\bf{k}=(k_1,k_2)\mapsto (\theta^1(\bf{k}),\theta^2(\bf{k}))=(k^1,k^2)$ to identify the two tori.

We have the following result and associated corollaries.\\

\newtheorem*{th:tomoki}{Theorem \ref{th:tomoki}}
\begin{th:tomoki} The following inequalities hold
\begin{align*}
\pi |\mathcal{C}|\leq \vol_g(\BZ^2)\leq \vol_{\widetilde{g}}(T^2_{\theta}),
\end{align*}
where the left-hand side inequality is saturated if and only if $\sqrt{\det(g)}=\frac{|\Omega_{12}|}{2}$ and provided $\Omega_{12}$ does not change sign over $\BZ^2$, and the right-hand side inequality is satisfied if and only if, for every $\bf{k}\in \BZ^2$,
\begin{align*}
g_{ij}(\bf{k})=e^{2f(\bf{k})}\widetilde{g}_{ij}(\theta=\bf{k}),\ 1\leq i,j\leq 2,
\end{align*}
for some function $f\in C^{\infty}(\BZ^2)$, i.e.,  $g$ is related to $\widetilde{g}$ by a Weyl rescaling, implying that, provided $P$ is an immersion, they share the same complex structure $j$ (and hence the same modular parameter).
\end{th:tomoki}

Observe that the topological invariant $\pi|\mathcal{C}|$ appears as the lower bound for both quantum volumes $\vol_{g}(\BZ^2)$ and $\vol_{\widetilde{g}}(T^2_{\theta})$. In other words, the first Chern number of the occupied Bloch bundle determines the minimal quantum volume. As a consequence of Theorem~\ref{th:tomoki}, we have the following corollaries.

\begin{corollary} If $P$ is an immersion, and thus $\det (g) \neq 0$ everywhere, then $\vol_{g}(\BZ^2)=\pi|\mathcal{C}|$ if and only if $P$ is K\"{a}hler by Theorem~\ref{th:kahlerband}.
\label{corol:1}
\end{corollary}

\begin{corollary} If $\widetilde{g}$ is non-degenerate and $\vol_{\widetilde{g}}(T^2_{\theta})=\pi|\mathcal{C}|$ then $(-\frac{1}{2}\widetilde{\Omega},\widetilde{g},\pm \widetilde{j})$ is a flat K\"{a}hler structure. The converse is also true. In both cases, we also have $\vol_{g}(\BZ^2)=\pi|\mathcal{C}|$. If in the first implication we have, in addition, that $P$ is an immersion, implying $g$ in non-degenerate, then $P$ is a K\"{a}hler map for the same complex structure. This implies that the same Weyl rescaling that relates $g$ and $\widetilde{g}$ in this case, relates also $\Omega$ and $\widetilde{\Omega}$, i.e.,
\begin{align*}
\Omega=e^{2f}\widetilde{\Omega}.
\end{align*}
Because the integral of the curvatures yields $\mathcal{C}$, we have $\int_{\BZ^2}\frac{d^2k}{(2\pi)^2} e^{2f(\bf{k})}=1$, which is consistent with the findings in the proof of Theorem~\ref{th:tomoki}.
\label{corol:2}
\end{corollary}

\begin{remark} Corollaries~\ref{corol:1} and~\ref{corol:2} are, in the mathematics literature, presented as a corollary of Wirtinger's inequality (see Proposition~\ref{prop:CS} and its proof in Appendix~\ref{Appendix:A}) which states that every complex submanifold of a K\"{a}hler manifold is volume minimizing in its homology class.
\end{remark}

\begin{corollary} If any of the volumes $\vol_{g}(\BZ^2)<\pi$ or $\vol_{\widetilde{g}}(T^2_{\theta})<\pi$ we have $\mathcal{C}=0$.
\label{corol:3}
\end{corollary}

Corollary~\ref{corol:3} is perhaps one of the most dramatic consequences in the sense that the volumes measured by the quantum metrics, which are purely geometric quantities, allow us to infer about non-trivial topological properties of the insulator.

The results presented place strong constraints in the geometry of Chern bands, which is relevant to the study of the stability of fractional topological insulators. Consider now, for the sake of simplicity of the following discussion, the case of a single band, i.e., $r=1$. In the works of Refs.~\cite{roy:14, jac:mol:roy:15}, the stability of topological phases arising in fractionally filled Chern insulators is studied. The saturation of the inequality in Proposition~\ref{prop:CS} together with the Fubini-Study metric being uniform throughout the Brillouin zone is then found as a criterion for the algebra of projected operators to be isomorphic to the one in the fractional quantum Hall effect --- the $W_{\infty}$ algebra found by Girvin, MacDonald and Platzman~\cite{gir:mac:pla:86}. Observe that this condition corresponds to the saturation of both inequalities in Theorem~\ref{th:tomoki} and, provided non degeneracy of the metric, it is the situation described in Corollary~\ref{corol:2}, supplemented with the additional condition that the conformal factor is trivial to have constancy of the metric in momentum space. Furthermore, in Ref.~\cite{jac:mol:roy:15} the saturation of the inequality of Proposition~\ref{prop:CS} for every $\bf{k}\in \BZ^2$ was identified as having a K\"{a}hler structure in the Brillouin zone. The saturation of the inequality is not enough to have a K\"{a}hler structure as one needs the immersion condition or, equivalently, $\det(g)\neq 0$ everywhere. We remark that the holomorphicity condition $Q\partial P/\partial\bar{z}=0$ found in Proposition~\ref{prop: Kahler sufficient conditions} is the momentum space form of the real space equation $Q_{\alpha}(x'-iy')P_\alpha=0$ found in Ref.~\cite{roy:14}, for the case when both inequalities in Theorem~\ref{th:tomoki} are satisfied and the complex structure in $\BZ^2$ is the same as the one in $T^2_{\theta}$ and completely determined by a single (multi-valued) complex coordinate of the form $z=k_1\pm \tau k_2$, for some modular parameter $\tau\in \mathcal{H}$ and the sign takes into account a possible orientation flip with respect to the standard one in the Brillouin zone. In Ref.~\cite{lee:cla:tho:17} and in related works~\cite{lee:tho:qi:13,lee:pap:tho:15,cla:lee:tho:qi:dev:15}, a K\"{a}hler structure where the associated complex structure is given by the flat isotropic one, i.e., determined by $z=k_1+\tau k_2$, with $\tau=i$, plays an important role in the construction of ideal fractional Chern insulator models.

Finally, we present a few results which we noticed as a consequence of analyzing the particular case of two-band models.

\newtheorem*{th:bruno}{Theorem \ref{th:bruno}}
\begin{th:bruno} For two-dimensional two-band models, there must exist a point in the Brillouin zone where $\det(g)=0$. In other words, the map $P:\BZ^2\to \Gr_{1}(\mathbb{C}^2)=\mathbb{C}P^1\cong S^2$ cannot be an immersion.
\end{th:bruno}

The next proposition is known in the literature and, as far as we know, the first to notice it were Yu-Quan Ma {\it et al}.~\cite{ma:gu:fan:liu:13,ma:20}.

\begin{proposition} For two-dimensional two-band models the inequality of Proposition~\ref{prop:CS} is saturated $\sqrt{\det(g)}=|\Omega_{12}|/2$.
\end{proposition}

Looking at the proof of Proposition~\ref{prop:CS} in the Appendix, this last result follows from the fact that the orthogonal complement projector $Q(\bf{k})=I-P(\bf{k})$ has rank $1$ in this case and saturation of the inequality is automatic. As a consequence of this proposition we have the two following corollaries.

\begin{corollary}
For two-dimensional two-band models, there must exist a point in the Brillouin zone where the Berry curvature vanishes.
\label{corol:4}
\end{corollary}

We point out that in Ref.~\cite{lee:cla:tho:17} it was noted that it is impossible to have a uniform Berry curvature for two-dimensional two-band models -- a manifestation of Theorem~\ref{th:bruno} and Corollary~\ref{corol:4}.

 \begin{corollary}
 For two-dimensional two-band models, if $\Omega_{12}$ or $-\Omega_{12}$ is everywhere nonnegative then $\vol_{g}(\BZ^2)=\pi |\mathcal{C}|$.
 \label{corol:5}
 \end{corollary}

 This implies that the quantum volume $\vol_g(\BZ^2)$ directly tells us about the topology of quantum states in this particular case.

\section{Physical interpretation of the K\"{a}hler map condition}
\label{sec: physical interpretation of the Kahler condition}
In this section, we provide a physical interpretation of the K\"{a}hler map condition appearing in Theorem~\ref{th:kahlerband}, namely, that the Fermi projector $P(\bf{k})=\Theta(E_{F}-H(\bf{k}))$, besides providing an immersion of the Brillouin zone $\BZ^2$ to the Grassmannian $\Gr_{r}(\mathbb{C}^n)$, satisfies
\begin{align}
Q\frac{\partial P}{\partial\overline{z}}=0,
\end{align}
in local complex coordinates $z$ coming from the complex structure $j$ (or $-j$, depending on the sign of the first Chern number, see the proof of Theorem~1 in Appendix~\ref{Appendix:A}) in the Brillouin zone. We will phrase this condition in terms of the many-particle ground state of the band insulator, which is obtained by filling the Bloch bands below the Fermi level. Recall that this state is, at the formal level, a Slater determinant over the momenta $\bf{k}$ in the Brillouin zone, where each momentum piece is itself a Slater determinant determined by the $r$-dimensional subspace $\text{Im}\left[P(\bf{k})\right]\subset \mathbb{C}^n$, where $n$ is the total number of bands in the system [see Eq.~\eqref{eq: Slater det}]. We may construct a linear independent set of local Bloch wave functions $\ket{v_{1,\bf{k}}},....,\ket{v_{r,\bf{k}}}$ spanning $\text{Im}\left[P(\bf{k})\right]\subset \mathbb{C}^n$, depending smoothly on $\bf{k}$, for $\bf{k}$ in a small open set, which we can write in terms of the canonical basis of $\mathbb{C}^n$,
\begin{align}
 \ket{v_{i,\bf{k}}}=\sum_{j=1}^{n}v_{i}^{j}(\bf{k})\ket{j}, \ i=1,...,r.
\end{align}
Associated to these vectors, we have operators which create them in the many-particle Fock space:
\begin{align}
\xi_{i\bf{k}}^{\dagger}=\sum_{j=1}^{n}v_{i}^{j}(\bf{k})c_{j\bf{k}}^{\dagger}, \ i=1,...,r,
\end{align}
where $c_{j\bf{k}}^{\dagger}$ denotes the original fermion creation operator, which creates a fermion with momentum $\bf{k}$ and internal degree of freedom $j$, with $j=1,..,n$. The $\bf{k}-$piece of the many-particle ground state is then determined by
\begin{align}
\prod_{i=1}^{r}\xi_{i\bf{k}}^{\dagger}\ket{0},
\end{align}
where $\ket{0}$ denotes the vacuum. This piece is only determined up to a gauge transformation. Namely, we are allowed to perform linear combinations among the states $\ket{v_{i,\bf{k}}}$, $i=1,...,r$, which do not change $\text{Im}\left[P(\bf{k})\right]$. This amounts to taking the $n\times r$ matrix $v(\bf{k})=\left[v_{i}^{j}(\bf{k})\right]_{1\leq i\leq r,1\leq j\leq n}$ and multiplying it on the right by a $k\times k$ invertible matrix $S(\bf{k})$:
\begin{align}
v(\bf{k})\longmapsto v(\bf{k}) S(\bf{k}),
\end{align}
and this can be done locally, in a neighbourhood of the considered momentum, in a smooth way. The effect of this gauge transformation is to transform the $\bf{k}$ piece of the many-particle ground state as
\begin{align}
\prod_{i=1}^{r}\xi_{i\bf{k}}^{\dagger}\ket{0}\longmapsto \det \left(S(\bf{k})\right)\prod_{i=1}^{r}\xi_{i\bf{k}}^{\dagger}\ket{0},
\end{align}
which clearly does not change the state. It is useful to write the projector $P(\bf{k})$ in terms of the matrix $v(\bf{k})$. To do this, note that from  $v(\bf{k})$ we can build a unitary gauge, i.e., a gauge $u(\bf{k})=[u_{i}^{j}(\bf{k})]_{1\leq i\leq r,1\leq j\leq n}$ for which the $r$ columns of the new matrix form an orthonormal basis of $\text{Im}\left[P(\bf{k})\right]$. In matrix notation, this is equivalent to the requirement $u^{\dagger}(\bf{k})u(\bf{k})=I_{r}$. Since this is a gauge transformation, this can be achieved by multiplying on the right by a matrix $S(\bf{k})$ and we have
\begin{align}
u(\bf{k})=v(\bf{k})S(\bf{k}), \text{ such that } u^{\dagger}(\bf{k})u(\bf{k})=I_r.
\end{align}
The unitarity condition above implies that the gauge transformation $S(\bf{k})$ has to satisfy
\begin{align}
S^{\dagger}v^{\dagger}vS=I_r\implies (SS^{\dagger})^{-1}=v^{\dagger}v,
\end{align}
where here and below, we drop the $\bf{k}$ dependence from the expressions just to make the formulas less cumbersome to read.
If we find one $S$ satisfying this equation, obviously, multiplying it on the right by a unitary $r\times r$ matrix will give another valid solution, which is just a manifestation of the fact that all we have done with this choice (coming from the Hilbert space inner product structure) was to reduce the full complex linear gauge group $\text{GL}(r;\mathbb{C})$ to the unitary gauge group $\text{U}(r)\subset \text{GL}(r;\mathbb{C})$. We can then write
\begin{align}
P=uu^{\dagger}=vSS^{\dagger}v^\dagger=v(v^\dagger v)^{-1}v^{\dagger}.
\end{align}
If we recall that the Berry connection, denoted $\nabla$, is determined by the projection of the exterior derivative, we can determine the one-form connection coefficients $A=[A_{i}^{j}]_{1\leq i,j\leq r}$ in the gauge provided by $v$:
\begin{align}
\nabla v= Pdv=v\left((v^{\dagger} v)^{-1} v^{\dagger}dv\right) \implies A=(v^{\dagger} v)^{-1} v^{\dagger}dv.
\end{align}
The reader should compare the above expression to the case in which the gauge is unitary, denoted $u$, for which the expression of the connection coefficients, because $u^\dagger u=I_r$, reduces to the familiar form
\begin{align}
A&=(u^\dagger u)^{-1}u^{\dagger}du=u^\dagger du\nonumber\\
&=\left[u_i^\dagger du_j\right]_{1\leq i,j\leq r}=\left[\bra{u_{i\bf{k}}}d\ket{u_{j\bf{k}}}\right]_{1\leq i,j\leq r},
\end{align}
where $u_i$, $i=1,\dots,r$, denote the columns of the $n\times r$ matrix $u$ which provide a basis of Bloch wave functions $\ket{u_{i\bf{k}}}$, $i=1,\dots,r$, spanning the $r$ Bloch bands below the Fermi level at momentum $\bf{k}$, namely $\text{Im}(P(\bf{k}))$.

Now equipped with the notion of holomorphic local coordinate, we have preferred local gauges which are holomorphic, namely, since can write
\begin{align}
A=A^{(0,1)}+ A^{(1,0)},
\end{align}
where $A^{(0,1)}$ only contains $d\bar{z}$, while $A^{(1,0)}$ only contains $dz$. A \emph{holomorphic gauge} is one in which $A^{(0,1)}=0$. By applying a gauge transformation, $v\mapsto vS$, we have
\begin{align}
A\longmapsto A'= S^{-1}AS +S^{-1}dS.
\end{align}
A holomorphic gauge can be found by solving,
\begin{align}
A'^{(0,1)}=S^{-1}A^{(0,1)}S +S^{-1}\frac{\partial S}{\partial \bar{z}}d\bar{z}=0.
\label{eq: holomorphic gauge}
\end{align}
In complex dimension one this equation has no obstruction to integrability, because there are no $(2,0)$-forms, and one can always find a holomorphic gauge (see Proposition~3.7 of Ref.~\cite{kob:87}). It is then clear the reason for calling it a holomorphic gauge: if we have any two holomorphic gauge choices, $v$ and $v'$, defined on the same local neighbourhood, then the gauge transformation $S$ relating them, $v'=vS$, is seen to be holomorphic by Eq.~\eqref{eq: holomorphic gauge} (set $A^{(0,1)}=A'^{(0,1)}=0$). Essentially, by choosing local holomorphic gauges, what we are doing is to equip our occupied Bloch vector bundle $E\to \BZ^2$ with the structure of a holomorphic vector bundle where the Cauchy-Riemann operator is simply the $(0,1)$ part of the Berry connection. Since the Berry connection preserves the Hermitian inner product, i.e., parallel transport is unitary, the Berry connection becomes the Chern connection of the Hermitian holomorphic vector bundle $E\to\BZ^2$(see Proposition~4.9 of~\cite{kob:87}).

In particular, in a holomorphic gauge $v$, we have
\begin{align}
A^{(0,1)}=\left(v^{\dagger}v\right)^{-1}v^{\dagger}\frac{\partial v}{\partial\bar{z}}d\bar{z}=0 \implies v^{\dagger}\frac{\partial v}{\partial \bar{z}}=0,
\end{align}
since $v^{\dagger}v$ is invertible. Now the condition  $Q\partial P/\partial \bar{z}=0$ means that actually, locally, we can choose a holomorphic gauge and that in such gauge $v=v(z)$, i.e., the local Bloch wave functions themselves can be chosen to be holomorphic functions. To see this, we will write the condition explicitly. We begin by taking $v$ to be a holomorphic gauge, meaning $v^{\dagger}\partial v/\partial \overline{z}=0$ and we will write $P$ in terms of $v$. Also, it will be more convenient to work with the equivalent condition $Q(\partial P/\partial \bar{z})P=0$ which follows due to $P$ being a projection operator (since $Q=I-P$ and $QdQQ=0$ we have $QdP=-QdQ=-QdQP=QdPP$). We then have,
\begin{align}
&Q\frac{\partial P}{\partial \bar z}P =Q\frac{\partial}{\partial\bar{z}}
\left(v(v^{\dagger}v)^{-1}v^{\dagger}\right)P \nonumber\\
&=Q\Big[\frac{\partial v}{\partial\bar{z}}(v^\dagger v)^{-1}v^{\dagger}-v(v^{\dagger}v)^{-1}\frac{\partial}{\partial\bar{z}}\left(v^{\dagger}v\right)(v^{\dagger}v)^{-1}v^{\dagger} \nonumber \\
&+v(v^{\dagger}v)^{-1}\frac{\partial v^{\dagger}}{\partial \bar{z}}\Big]P\nonumber\\
&=Q\Big[\frac{\partial v}{\partial\bar{z}}(v^\dagger v)^{-1}v^{\dagger}-v(v^{\dagger}v)^{-1}\frac{\partial v^{\dagger}}{\partial\bar{z}}v(v^{\dagger}v)^{-1}v^{\dagger} \nonumber \\
&-v(v^{\dagger}v)^{-1}v^{\dagger}\frac{\partial v}{\partial\bar{z}}v(v^{\dagger}v)^{-1}v^{\dagger}+v(v^{\dagger}v)^{-1}\frac{\partial v^{\dagger}}{\partial \bar{z}}\Big]P\nonumber\\
&=Q\Big[\frac{\partial v}{\partial\bar{z}}(v^\dagger v)^{-1}v^{\dagger}-v(v^{\dagger}v)^{-1}\frac{\partial v^{\dagger}}{\partial\bar{z}}v(v^{\dagger}v)^{-1}v^{\dagger} \nonumber \\
&+v(v^{\dagger}v)^{-1}\frac{\partial v^{\dagger}}{\partial \bar{z}}\Big]P\nonumber \\
&=Q\Big[\frac{\partial v}{\partial\bar{z}}(v^\dagger v)^{-1}v^{\dagger}-v(v^{\dagger}v)^{-1}\frac{\partial v^{\dagger}}{\partial\bar{z}} \nonumber \\
&+v(v^{\dagger}v)^{-1}\frac{\partial v^{\dagger}}{\partial \bar{z}}\Big]P=Q\left[\frac{\partial v}{\partial\bar{z}}(v^\dagger v)^{-1}v^{\dagger}\right]P=0,
\end{align}
where we used the holomorphic gauge condition and the fact that $P=v(v^\dagger v)^{-1}v^{\dagger}$ and it is a projection operator, i.e., $P^2=P$. Since $Pv=v(v^{\dagger}v)^{-1}v^{\dagger}v=v$ it follows that $v^{\dagger}P=v^{\dagger}$, so we can write the previous condition as
\begin{align}
Q\frac{\partial v}{\partial \bar{z}}(v^\dagger v)^{-1}v^{\dagger}=0.
\end{align}
Multiplying on the right by $v$, we get
\begin{align}
Q\frac{\partial v}{\partial \bar{z}}=0.
\end{align}
Now since from the holomorphic gauge constraint we have $A^{(0,1)}=0\Leftrightarrow \nabla_{\frac{\partial}{\partial \bar{z}}} v=P\partial v/\partial\bar{z}=0$, it follows that $v$ satisfies $\partial v/\partial \bar{z}=0$, i.e. it is holomorphic. We conclude that the ground state of the band insulator can be locally described by local holomorphic Bloch wave functions $\ket{v_{1\bf{k}}},\dots,\ket{v_{r\bf{k}}}$. This should be compared to what happens in the integer quantum Hall effect in which the ground state wave function is, in the symmetric gauge, a Gaussian factor times a holomorphic Slater determinant of single-particle states taken from the lowest Landau level. The Gaussian factor in the integer Hall effect appears because the symmetric gauge is unitary, but one could also work with a holomorphic gauge, in which case that factor does not appear and, effectively, the single-particle states from the lowest Landau level in the plane are just holomorphic functions which are square integrable with respect to an $L^2-$inner product. In analogy with what happens in the lowest Landau level, here the local wave functions can be chosen to effectively only depend on half of the degrees of freedom, meaning they only depend on the $z$ coordinate and not on $\bar{z}$. An alternative way of understanding this feature for the integer quantum Hall effect is to use the Landau gauge, in which the lowest Landau level is identified with the Hilbert space of square integrable functions in the momentum variable with respect to which translation symmetry is preserved and, hence, it corresponds to the Hilbert space of a particle in dimension $1$. The sign of the first Chern number $\mathcal{C}$ associated with the Fermi projector, by determining whether the map is K\"{a}hler with respect to $j$ or $-j$ (see the proof of Theorem~1 in Appendix~\ref{Appendix:A}), controls the notion of orientation or chirality of the wave function in momentum space, for if we locally write $z=|z|e^{i\theta}$ any holomorphic function will wind in the positive orientation of the $z$ plane, i.e., in the positive $\theta$ direction. Alternatively and in a complementary way, we can also connect this notion of chirality to the sign of the first Chern number by considering parallel transport with respect to the Berry connection around loops. The fact that $P^*\omega_{FS}$ is a symplectic form when the map is an immersion has as a consequence that $P^*\omega_{FS}$ is a volume form on the Brillouin zone. This volume form (because it is nowhere vanishing) determines an orientation on the Brillouin zone and, because the map is K\"{a}hler, it is also the one induced by the complex structure $j$ (or $-j$) which makes the map $P:\BZ^2\to \Gr_{k}(\mathbb{C}^n)$ holomorphic. The fact that $P^*\omega_{FS}=-\Omega/2$, where $\Omega$ is the Berry curvature, then tells us that in this orientation the Chern number is negative (see also the proof of Theorem~\ref{th:kahlerband} in Appendix~\ref{Appendix:A}). Suppose for a moment that we are considering a single band, i.e. $r=1$, for the sake of simplicity (equivalently, consider the top exterior power line bundle $\Lambda^r E\to \BZ^2$, whose fibers describe Slater determinants of vectors in the corresponding fibers of the occupied Bloch vector bundle $E\to \BZ^2$). If $\gamma$ is a closed loop which is the boundary of some region $\Sigma\subset \BZ^2$, with $\Sigma$ having the same orientation as $\BZ^2$, we will have, by Stokes' theorem,
\begin{align}
\exp\left(i\int_{\gamma} A\right)=\exp\left(i\int_{\Sigma} \Omega\right)=\exp\left(-2i\int_{\Sigma} P^*\omega_{FS} \right),
\end{align}
and because $P^*\omega_{FS}$ is the volume form,
\begin{align}
\int_{\Sigma} P^*\omega_{FS}=\int_{\Sigma} \omega > 0
\end{align}
so that the argument of the parallel transport phase, $\int_{\Sigma} \Omega$, which equals, in magnitude, twice the symplectic volume of $\Sigma$ measured with respect to $\omega$, or equivalently, due to $P$ being K\"{a}hler, twice the volume of $\Sigma$ measured with respect to the quantum metric $g$, will necessarily have negative sign. We notice also that the statement is gauge invariant and is valid for every loop $\gamma$ enclosing some arbitrary region $\Sigma\subset \BZ^2$ in the same orientation as the Brillouin zone.

The phenomenon of chirality described above is also similar to what happens in the lowest Landau level, in which the sign of the magnetic field also controls the chirality of the wave function. The fact that the geometrical degrees of freedom are cut in half is, intuitively, and, in fact, formally due to Theorem~\ref{th:kahlerband}, also consistent with the saturation of the Cauchy-Schwarz like inequality in Theorem~\ref{th:kahlerband} in that this band insulator will have less geometrical degrees of freedom in its description because the quantum metric and the Berry curvature are no longer independent quantities.

We point out that a closely related result has been obtained also by Lee \emph{et al.}~\cite{lee:cla:tho:17}; they have shown, in the context of a single band, that the (local) Bloch wave function is a holomorphic function with respect to the isotropic flat complex structure, $\tau = i$, if and only if the inequality $\mathrm{tr}(g) = 2\sqrt{\det (g)} \ge |\Omega_{12}|$ is saturated. Our result holds more generally for a general complex structure given by Eq.~(\ref{eq:j}), also including a general number of bands $r$. Besides, the results presented here re-express this local holomorphicity condition (expressed in terms of local Bloch wave functions) as a global holomorphicity condition for the map to the Grassmannian $\Gr_{r}(\mathbb{C}^n)$ induced by the Fermi projector $P(\bf{k})=\Theta(E_F-H(\bf{k}))$, which, provided the additional immersion condition, $\det (g)\neq 0$ everywhere in $\BZ^2$, is translated into a K\"{a}hler map condition.

\section{Conclusions}
\label{sec: conclusions}
In this paper we have studied the geometry of Chern insulators using the natural K\"{a}hler geometry of the space of orthogonal projectors of a given rank, the latter being associated with the number of occupied bands. In the process, we have established important results, namely Theorem~\ref{th:kahlerband}, which renders the saturation of the Cauchy-Schwarz-like inequality equivalent to a K\"{a}hler map condition for the Fermi projector, and Theorem~\ref{th:tomoki}, which determines the conditions of minimal quantum volume both in the Brillouin zone and in the twist-angle space. The minimal quantum volume is, in both cases, precisely given by the topological invariant $\pi|\mathcal{C}|$, where $\mathcal{C}$ is the first Chern number of the occupied bands. Supplemented with a non-degeneracy condition, the minimal volume conditions amount to these tori being endowed with the structure of a K\"{a}hler manifold inherited from the K\"{a}hler geometry of the space of quantum states. Strikingly, for the case of a single Chern band, when both quantum volumes are equal and minimal, together with the additional condition that the conformal factor relating the metrics is trivial, we get that the algebra generated by the projected density operators is isomorphic to the $W_{\infty}$ algebra in the fractional Hall effect~\cite{roy:14}. Hence, this can also be used as a necessary criterion for the stability of fractional topological insulators.

We have related the K\"{a}hler map condition appearing in Theorem~\ref{th:kahlerband} to the possibility of writing local holomorphic Bloch functions, thus obtaining a description of the many-particle ground state of the insulator which is quite close to the one in the integer quantum Hall effect in which, apart from an overall Gaussian factor, we have a Slater determinant of holomorphic functions. The sign of the Chern number is seen to control the ``chirality'' of the wave function just like the sign of the magnetic field does in the case of the quantum Hall effect.

A remarkable consequence of Theorem~\ref{th:tomoki} is Corollary~\ref{corol:3} which states that if any of the quantum volumes is smaller than $\pi$ then we immediately have that $\mathcal{C}=0$, i.e., we have a trivial Chern insulator. In the accompanying paper~\cite{oz:mer:21:published}, it is also shown with several tight-binding models that even when $\mathcal{C} \neq 0$, the quantum volume often gives a good estimate of the Chern number. This shows how the quantum volumes, intrinsically geometric quantities, are tied to the non-trivial topology of the Chern insulator.

Additionally, we found that, for two-band systems, it is impossible to satisfy the non-degeneracy condition due to topological constraints placed by the fundamental group of the two-torus. However, provided the Berry curvature is everywhere non-negative or nonpositive, the quantum volume will automatically reach the minimal value $\pi|\mathcal{C}|$, the reason being due to the minimal dimension of the total number of internal degrees of freedom.
\begin{acknowledgments}
B.M. acknowledges very stimulating discussions with J. P. Nunes and J. M. Mour\~{a}o. B.M. and T.O. acknowledge the anonymous referees for their insightful comments which lead to substantial improvement of this article. TO acknowledges support from JSPS KAKENHI Grant Number JP20H01845, JST PRESTO Grant Number JPMJPR19L2, JST CREST Grant Number JPMJCR19T1, and RIKEN iTHEMS. BM acknowledges the support from SQIG -- Security and Quantum Information Group, the Instituto de Telecomunica\c{c}\~oes (IT) Research Unit, Ref. UIDB/50008/2020, funded by Funda\c{c}\~ao para a Ci\^{e}ncia e a Tecnologia (FCT), European funds, namely, H2020 project SPARTA, as well as  projects QuantMining POCI-01-0145-FEDER-031826 and PREDICT PTDC/CCI-CIF/29877/2017.
\end{acknowledgments}

\appendix
\section{Detailed mathematical proofs}
\label{Appendix:A}
\begin{proposition} If $P:\Sigma\to \Gr_{r}(\mathbb{C}^n)$ is an immersion of a Riemann surface to the Grassmannian which is holomorphic, i.e., it satisfies, in every local complex $z$ coordinate on the surface,
\begin{align*}
Q\frac{\partial P}{\partial \bar{z}}=0,
\end{align*}
then it is a K\"{a}hler map with respect to the usual K\"{a}hler structure on $\Gr_{r}(\mathbb{C}^n)$ given by $(\omega_{FS},J_{FS}, g_{FS})$. The converse is  also true, i.e., if $P:\Sigma\to\Gr_r(\mathbb{C}^n)$ is K\"{a}hler then it is an holomorphic immersion of $\Sigma$.
\label{prop: Kahler sufficient conditions}
\end{proposition}

\begin{proof}
Since $P$ is an immersion, we will immediately have that $g=P^*g_{FS}$ is a Riemannian metric. Furthermore, the two-form $\omega=P^*\omega_{FS}$ is automatically closed. It remains to show that if $P$ is holomorphic then $\omega(\cdot,j\cdot)=g$ and also that $\omega$ is non-degenerate. Recall that $P$ being holomorphic means that the differential of $P$ intertwines the two complex structures (see, for example, Ref.~\cite{huy:05}), i.e.,
\begin{align}
J_{FS}\circ dP=dP\circ j.
\end{align}
If we take a complex coordinate $z$ associated to $j$ over $\Sigma$, so that
\begin{align}
j\left(\frac{\partial}{\partial z}\right)=i\frac{\partial}{\partial z} \text{ and } j\left(\frac{\partial}{\partial \overline{z}}\right)=-i\frac{\partial}{\partial \overline{z}},
\label{eq: j and z}
\end{align}
we have
\begin{align}
J_{FS}\left(\frac{\partial P}{\partial z}\right)=i\frac{\partial P}{\partial z}.
\end{align}
Observing that $dP=QdPP +PdPQ$, that $\left(QdPP\right)\circ J_{FS}=iQdPP$ and that $\left(PdPQ\right)\circ J_{FS}=-iPdPQ$, we get,
\begin{align}
iQ\frac{\partial P}{\partial z}P -iP\frac{\partial P}{\partial z}Q =iQ\frac{\partial P}{\partial z}P +iP\frac{\partial P}{\partial z}Q.
\end{align}
This then implies the equation
\begin{align}
P\frac{\partial P}{\partial z}Q=0,
\end{align}
or, equivalently, the Hermitian conjugate
\begin{align}
Q\frac{\partial P}{\partial \overline{z}} P=Q\frac{\partial P}{\partial \overline{z}}=0,
\end{align}
where the last equality follows from $QdPQ=Q(-dQ)Q=0$.
Observe that
\begin{align*}
&P^*\omega_{FS}=-\frac{i}{2}\tr\left(PdP\wedge dP\right)\\
&=-\frac{i}{2}\left(\tr\left(P\frac{\partial P}{\partial \bar{z}}\frac{\partial P}{\partial z}\right)-\tr\left(P\frac{\partial P}{\partial z}\frac{\partial P}{\partial \bar{z}}\right)\right)d\bar{z}\wedge dz\\
&=-\frac{i}{2}\left(\tr\left(P\frac{\partial P}{\partial \bar{z}}Q\frac{\partial P}{\partial z}\right)-\tr\left(P\frac{\partial P}{\partial z}Q\frac{\partial P}{\partial \bar{z}}\right)\right)d\bar{z}\wedge dz\\
&=-\frac{i}{2}\tr\left(P\frac{\partial P}{\partial \bar{z}}Q\frac{\partial P}{\partial z}\right)d\bar{z}\wedge dz=\frac{i}{2}\tr\left(P\frac{\partial P}{\partial \bar{z}}\frac{\partial P}{\partial z}\right)dz\wedge d\bar{z},
\end{align*}
where we used $PdPP=0$ and $QdQQ=0$ which imply that $PdP=PdPQ=dPQ=-dQQ$ and similarly $dPP=QdPP=QdP$. Also, by the same reasoning
\begin{align*}
P^*g_{FS}&=\tr\left(PdPdP\right)=\tr\left(P\frac{\partial P}{\partial \bar{z}}\frac{\partial P}{\partial z}\right)d\bar{z} dz.
\end{align*}
Using Eq.~\eqref{eq: j and z} we get
\begin{align*}
\frac{i}{2}dz\wedge d\bar{z}(\cdot ,j\cdot)&=\frac{i}{2}\left(dz\otimes d\bar{z}-d\bar{z}\otimes dz\right)(\cdot,j\cdot)\\
&=\frac{1}{2}\left(dz\otimes d\bar{z} +d\bar{z}\otimes dz\right)=dzd\bar{z}.
\end{align*}
Thus, we conclude
\begin{align*}
\omega(\cdot,j\cdot)=P^*\omega_{FS}(\cdot,j\cdot)=P^*g_{FS}=g.
\end{align*}
Finally, $\omega$ is non-degenerate because $g$ is so and they are both determined by the same function, namely $\tr\left(P\frac{\partial P}{\partial \overline{z}}\frac{\partial P}{\partial z}\right)$.

The converse is also true, for if $P:\Sigma\to\Gr_r(\mathbb{C}^n)$ is K\"{a}hler it automatically is an immersion since that is implied by $g$ being non-degenerate and it has to be holomorphic in order to be K\"{a}hler. Hence, the proof is concluded.
\end{proof}

\begin{proposition}[Cauchy-Schwarz-type inequality and relation to the Wirtinger inequality]
\label{prop:CS}
We have the inequality
\begin{align*}
\sqrt{\det [g_{ij}(\bf{k})]}\geq \frac{|\Omega_{12}(\bf{k})|}{2},
\end{align*}
with the equality, for a given $\bf{k}\in \BZ^2$, given by at least one of the following conditions
\begin{itemize}
\item [(i)] $Q(\bf{k})\frac{\partial P}{\partial k_1}(\bf{k})=0$,
\item [(ii)]$Q(\bf{k})\frac{\partial P}{\partial k_2}(\bf{k})=0$,
\item [(iii)] $\exists \lambda\in \mathbb{C}-\{0\}: Q(\bf{k})\frac{\partial P}{\partial k_1}(\bf{k})=\lambda Q(\bf{k})\frac{\partial P}{\partial k_2}(\bf{k})$.
\end{itemize}
The above result is a manifestation of the so-called Wirtinger inequality for general K\"{a}hler manifolds~\cite{ban:kat:shn:wei:09}. The Wirtinger inequality states that on a K\"{a}hler manifold, the $k$th exterior power of the K\"{a}hler form, when evaluated on a simple (decomposable) 2$k$-vector of unit volume, is bounded above by $k!$. The case at hand follows with $k=1$.
\end{proposition}

\begin{proof}
The inequality has been first noted by Roy~\cite{roy:14}, for the case $r=1$. A physicist-friendly proof, for $r\geq 1$, can be found in the accompanying paper~\cite{oz:mer:21:published}. Below we give an alternative proof using the Wirtinger inequality. Before doing that, we give a few additional remarks here to the proof of Ref.~\cite{oz:mer:21:published} because the conditions for the saturation of the inequality are stated here in terms of projectors. We first observe that conditions (i), (ii), (iii) are statements about the image of the differential of $P$ at $\bf{k}$ because $d_{\bf{k}}P(\partial/\partial k_i)$ is given by $Q(\bf{k})\left(\partial P/\partial k_i\right)(\bf{k})+ \text{H.c.}$, $i=1,2$. Moreover, the fact that the Pl\"{u}cker embedding is an isometry allows us to reduce the proof of the proposition to the case of $r=1$, i.e., for the complex projective space $\mathbb{P}(\Lambda^r\mathbb{C}^n)\cong \mathbb{C}P^{m}=\Gr_{1}(\mathbb{C}^{m+1})$ with $m=\binom{n}{r}-1$. Effectively, this amounts to, for each $\bf{k}\in\BZ^2$, replacing the initial rank $r$ projector by the rank $1$ projector associated to the wedge product (Slater determinant) associated to an orthonormal basis of the image of the former. In this case, we have a rank $1$ projector locally specified by a (normalized) Bloch wave function $P(\bf{k})=\ket{\psi(\bf{k})}\bra{\psi(\bf{k})}$.  It follows from the proof in Ref.~\cite{oz:mer:21:published} that the inequality is saturated if and only if $Q(\bf{k})\partial/\partial k_1\ket{\psi(\bf{k})}=0$ or $Q(\bf{k})\partial/\partial k_2\ket{\psi(\bf{k})}=0$ or $Q(\bf{k})\partial/\partial k_1\ket{\psi(\bf{k})}=\lambda Q(\bf{k})\partial/\partial k_2\ket{\psi(\bf{k})}=0$ for some $\lambda\in\mathbb{C}$. Additionally, observing that
\begin{align}
Q(\bf{k})\frac{\partial P}{\partial k_i}(\bf{k})=Q(\bf{k})\left(\frac{\partial}{\partial k_i}\ket{\psi(\bf{k})}\right)\bra{\psi(\bf{k})}, \text{ for } i=1,2,
\end{align}
because $Q(\bf{k})\ket{\psi(\bf{k})}=0$, the proof is complete.

Now we consider the relation with the Wirtinger inequality, which provides an alternative proof. For $k=1$, in our particular setting, this means that if $e_1,e_2$ are orthonormal with respect to $g_{FS}$, we have
\begin{align}
|\omega_{FS}(e_1,e_2)|\leq 1
\end{align}
If $e_1$ and $e_2$ form an orthonormal basis for the image of the differential $dP$ for some momentum $\bf{k}$, then,
\begin{align}
e_{i}=\sum_{j=1}^2a_{i}^{j} dP\left(\frac{\partial}{\partial k_j}\right), i=1,2
\label{eq: o.n. basis of the image of dP}
\end{align}
for some invertible matrix $a=[a_{i}^{j}]_{1\leq i,j\leq 2}$. Then,
\begin{align*}
\omega_{FS}(e_1,e_2)&=\det (a_{i}^{j}) \omega_{FS} (dP\left(\frac{\partial}{\partial k_1}\right),
dP\left(\frac{\partial}{\partial k_1}\right))\\
&= \det(a_{i}^{j})\left(P^{*}\omega_{FS}\right)_{12}=-\frac{1}{2}\Omega_{12} \det(a_{i}^{j}).
\end{align*}
Next, observe that since $g_{FS}(e_i,e_j)=\delta_{ij}$, $1\leq i,j\leq 2$, we have that
\begin{align*}
g_{ij}=(P^*g_{FS})\left(\frac{\partial}{\partial k_i}, \frac{\partial}{\partial k_j}\right)=\sum_{k,l=1}^2(a^{-1})^{k}_{i}
(a^{-1})^{l}_{j}\delta_{kl}.
\end{align*}
As a consequence, $\det(g_{ij})=\det(a^{-1})^2$. Thus, the Wirtinger inequality implies
\begin{align}
|\omega_{FS}(e_1,e_2)|=\left|-\frac{1}{2}\Omega_{12} {\det}(a_{i}^{j})\right|\leq 1
\end{align}
or
\begin{align}
  \sqrt{\det(g)}\geq \frac{1}{2}|\Omega_{12}|,
\end{align}
as claimed. Concerning the saturation of the inequality, two cases are to be considered in light of Wirtinger's inequality. If at a given point $\bf{k}\in\BZ^2$ the mapping is not an immersion the dimension of image of the differential of $P$ at $\bf{k}$ is smaller than $2$. This implies that $\sqrt{\det[g_{ij}(\bf{k})]}=\Omega_{12}(\bf{k})=0$ at that point because we can not build an orthonormal basis with two elements in the image. This covers the cases (i) and (ii), and also (iii) provided $\lambda$ is real since they imply that the differential of $P$ at $\bf{k}$ is not injective (the vanishing of $\sum_{i=1}^{2} a^iQ(\bf{k})\partial P/\partial k_i(\bf{k})$ is equivalent to the vanishing of the image under the differential of $P$ at $\bf{k}$ of the tangent vector $\sum_{i=1}^{2}a^i\partial/\partial k_i\in T_{\bf{k}}\BZ^2$ with $a^i\in\mathbb{R}$, $i=1,2$). We are left with the non-degenerate case, i.e., when the image of the differential of $P$ at $\bf{k}$ is two-dimensional, in which we can apply Wirtinger's inequality and consider the situation in which it is saturated. We will show in the following that this is equivalent to the case (iii) with $\text{Im}\left(\lambda\right)\neq 0$. Because $(\omega_{FS},J_{FS},g_{FS})$ is a compatible triple, the saturation of Wirtinger's inequality holds iff $e_2=\pm J_{FS}(e_1)$. Indeed, if $e_{2}=\pm J_{FS}(e_1)$ we have $|\omega_{FS}(e_1,e_2)|=g_{FS}(e_1,e_1)=1$. Conversely if $|\omega_{FS}(e_1,e_2)|=1$, we have
\begin{align}
|\omega_{FS}(e_1,e_2)|&=|g(J_{FS}(e_1),e_2)|\nonumber \\
&=1\nonumber \\
&=g_{FS}(e_1,e_1)=g_{FS}(e_2,e_2).
\end{align}
Furthermore, by skew-symmetry, $\omega_{FS}(e_i,e_i)=0=g_{FS}(e_i,J(e_i))$, $i=1,2$ and, by compatibility, $g_{FS}(J_{FS}(e_1),J_{FS}(e_1))=\omega_{FS}(e_1,J_{FS}(e_1))=g(e_1,e_1)$. It follows that $J_{FS}(e_1)$ is orthogonal to $e_1$ and it has norm $1$. From the saturation of the Cauchy-Schwarz inequality for the metric $g_{FS}$
\begin{align}
1\!=\!g_{FS}(e_2,e_2)g_{FS}(J_{FS}(e_1),J_{FS}(e_1))&\geq
|g_{FS}(e_2,J_{FS}(e_1))|^2\nonumber\\
&=|\omega_{FS}(e_1,e_2)|^2\nonumber \\
&=1,
\end{align}
it follows that $J(e_1)=\pm e_2$. Now if $e_1,e_2$ form an orthonormal basis of the image of $dP$ at $\bf{k}$ as in Eq.~\eqref{eq: o.n. basis of the image of dP} this in turn implies that
\begin{align}
\sum_{j=1}^{2}a_{1}^{j}\left(iQ\frac{\partial P}{\partial k_j}-i\frac{\partial P}{\partial k_j}Q\right)=\pm\sum_{j=1}^{2}a_{2}^{j}\left(Q\frac{\partial P}{\partial k_j}+\frac{\partial P}{\partial k_j}Q\right).
\end{align}
Due to $P$ and $Q$ being orthogonal to each other it follows that
\begin{align}
\sum_{j=1}^{2}a_{1}^{j}iQ\frac{\partial P}{\partial k_j}=\pm\sum_{j=1}^{2}a_{2}^{j}Q\frac{\partial P}{\partial k_j},
\end{align}
which further implies
\begin{align}
\left(ia_1^1\mp a_2^1\right)Q\frac{\partial P}{\partial k_1}=\left(-ia_2^1\pm a_2^2\right)Q\frac{\partial P}{\partial k_2}.
\end{align}
Now computing
\begin{align}
\lambda=\frac{(-i a^1_2 \pm a^2_2)}{(ia^1_1\mp a^1_2)}&=\frac{(-i a^1_2 \pm a^2_2)( -ia^1_1\mp a^1_2)}{\left((a^1_2)^2 +(a^1_1)^2\right)}\nonumber \\
&=\frac{-a^1_1a^1_2 -a^1_2a^2_2 +i\left(\mp a^1_1a^2_2\pm a^1_2a^1_2\right)}{\left((a^1_2)^2 +(a^1_1)^2\right)}\nonumber \\
&=\frac{-a^1_1a^1_2 -a^1_2a^2_2 \pm i \det(a^{i}_{j})}{\left((a^1_2)^2 +(a^1_1)^2\right)},
\end{align}
which, because $\det(a^i_j)\neq 0$ and $(a^1_2)^2 +(a^1_1)^2\neq 0$ (since $e_2\neq 0$), implies condition (iii) with $\mbox{Im}(\lambda)\neq 0$. This concludes our proof.
\end{proof}

\begin{theorem} Suppose $P:\BZ^2\to\Gr_{r}(\mathbb{C}^n)$ is an immersion, then the Cauchy-Schwarz like inequality is saturated
\begin{align*}
\sqrt{\det(g_{ij}(\bf{k}))}=\frac{\mp \Omega_{12}(\bf{k})}{2}, \text{ for all } \bf{k}\in \BZ^2,
\end{align*}
if and only if the map $P$ is K\"{a}hler with respect to the triple of structures $(\omega, \pm j, g)$.
\label{th:kahlerband}
\end{theorem}

\begin{proof}
By assumption, $P$ is an immersion so the equality
\begin{align*}
\sqrt{\det (g_{ij}(\bf{k}))}=\frac{|\Omega_{12}(\bf{k})|}{2},
\end{align*}
implies that condition (iii) of Proposition~\ref{prop:CS} holds pointwise, i.e., there exists a function $\lambda\in C^{\infty}(\BZ^2)$ such that
\begin{align*}
Q(\bf{k})\frac{\partial P}{\partial k_1}(\bf{k})=\lambda(\bf{k}) Q(\bf{k})\frac{\partial P}{\partial k_2}(\bf{k}), \text{ for all } \bf{k}\in\BZ^2.
\end{align*}
Observe furthermore that by the fact that $P$ is an immersion we have that $g$ is non-degenerate and hence $\mbox{Im}[\lambda(\bf{k})]\neq 0$ for all $\bf{k}\in\mathbb{R}$. Otherwise the quantities $Q\partial P/\partial k_i$, $i=1,2$, which represent the images of the tangent vectors $\partial/\partial k_i$, $i=,1,,2$, would be linearly dependent.

Observe that if we introduce a complex variable $z$ such that
\begin{align*}
\frac{\partial}{\partial \overline{z}}=\frac{1}{2}\left(\frac{\partial}{\partial k_1}-\lambda\frac{\partial}{\partial k_2}\right),
\end{align*}
the statement is
\begin{align*}
Q\frac{\partial P}{\partial \overline{z}}=0.
\end{align*}
To find $z$ we would have to solve,
\begin{align*}
\frac{\partial k_1}{\partial \bar{z}}=\frac{1}{2} \text{ and } \frac{\partial k_2}{\partial \bar{z}}=-\frac{\lambda}{2},\\
\frac{\partial k_1}{\partial z}=\frac{1}{2} \text{ and } \frac{\partial k_2}{\partial z}=-\frac{\overline{\lambda}}{2}.
\end{align*}
Equivalently, since
\begin{align*}
\left(\frac{1}{2}\left[\begin{array}{cc}
1 & 1\\
-\bar{\lambda} &-\lambda
\end{array}\right]\right)^{-1}=\frac{1}{i\text{Im}(\lambda)}\left[\begin{array}{cc}
\lambda & 1\\
-\bar{\lambda} & -1
\end{array}\right],
\end{align*}
well defined since $\mbox{Im}(\lambda)\neq 0$, we have
\begin{align*}
dz\!=\!\frac{1}{i\text{Im}(\lambda)}\!\left(\lambda dk_1 +dk_2\right)\! \text{ and }  d\overline{z}\!=\!-\frac{1}{i\text{Im}(\lambda)}\!\left(\overline{\lambda} dk_1 +dk_2\right)\!.
\end{align*}
The resulting partial differential equations are locally integrable (see, for instance, Example~5 of Chapter~1 in Ref.~\cite{che:67} or Theorem~4.16 of Ref.~\cite{mcd:sal:17}) because by solving them we are finding local isothermal coordinates, in two dimensions, with respect to $P^*g_{FS}$. These local systems of coordinates then glue together holomorphically giving the Brillouin zone the structure of a complex manifold. We now refer to Proposition~\ref{prop: Kahler sufficient conditions} and one side of the implication is proved. Note, however, that the new local complex coordinates $z$ can induce the opposite orientation with respect to the usual one induced by the (multi-valued) complex coordinate $k_1+ik_2$. This is the case when the 1st Chern number is positive. In this case, the map is K\"{a}hler with respect to the complex structure one gets by rotation of $90$ degrees as determined by $g$ and the orientation opposite to the standard orientation of the Brillouin zone; namely the almost complex structure is $-j$, with respect to $j$ as defined in Sec.~\ref{sec: the geometry of chern insulators}.
This follows from the fact that the function $\tr\left(P\frac{\partial P}{\partial \overline{z}}\frac{\partial P}{\partial z}\right)$ determining $P^*\omega_{FS}=-P^*\Omega_0/2=-\Omega/2$ in the local holomorphic coordinate $z$ is always positive (because it also determines the metric) and so the integral in an orientation consistent with it will be positive.

Let us now consider the converse, i.e., if $P:\BZ^2\to\Gr_{r}(\mathbb{C}^n)$ is K\"ahler (which automatically implies that $P$ is an immersion) with respect to the triple of structures $(\omega,\pm j,g)$, then the inequality
\begin{align*}
\sqrt{\det(g_{ij}(\bf{k}))}=\frac{|\Omega_{12}(\bf{k})|}{2}, \text{ for all } \bf{k}\in \BZ^2,
\end{align*}
is saturated. The reason is that the K\"{a}hler condition $\omega(\cdot ,\pm j\cdot)=g(\cdot,\cdot)$, with $\omega=P^*\omega_{FS}$ and $g=P^*g_{FS}$ is locally given by
\begin{align*}
\mp\sum_{j=1}^2\frac{1}{2}\Omega_{ij}(\bf{k})j^{j}_{k}(\bf{k})=g_{ik}(\bf{k}),
\end{align*}
with $j^{i}_{j}(\bf{k})$ being the matrix elements of the complex structure $j$, i.e.,
\begin{align*}
j\left(\frac{\partial}{\partial k_i}\right)=\sum_{j=1}^2 j^{j}_{i}\frac{\partial}{\partial k_j}, \ i=1,2.
\end{align*}
In matrix form
\begin{align*}
\pm\omega j=\mp\frac{1}{2}\Omega j=g.
\end{align*}
Taking determinants, we get
\begin{align*}
\frac{1}{4}\Omega_{12}^2\det(j)=\frac{1}{4}|\Omega_{12}|^2\det(j)=\det(g).
\end{align*}
Since the right-hand side has positive determinant it follows that $\det(j)>1$ and hence, because $j^2=-I$ we have $\det(j)^2=1$ so $\det(j)=1$. Hence we get the equality
\begin{align*}
\frac{1}{2}|\Omega_{12}|=\sqrt{\det(g)},
\end{align*}
as desired. The theorem is, thus, proved.
\end{proof}

\begin{theorem} The following inequalities hold
\begin{align*}
\pi |\mathcal{C}|\leq \mbox{vol}_g(\BZ^2)\leq \mbox{vol}_{\widetilde{g}}(T^2_{\theta}),
\end{align*}
where the left-hand side inequality is saturated iff, for every $\bf{k}\in \BZ^2$, $\sqrt{\det(g)}=\frac{|\Omega_{12}|}{2}$ and provided $\Omega_{12}$ does not change sign over $\BZ^2$, and the right-hand side inequality is satisfied iff
\begin{align*}
g_{ij}(\bf{k})=e^{2f(\bf{k})}\widetilde{g}_{ij}(\theta=\bf{k}),\ 1\leq i,j\leq 2,
\end{align*}
for some function $f\in C^{\infty}(\BZ^2)$, i.e., $g$ is related to $\widetilde{g}$ by a Weyl rescaling, implying that, provided $P$ is an immersion, they share the same complex structure $j$ (and hence the same modular parameter).
\label{th:tomoki}
\end{theorem}

\begin{proof}
The left-hand side inequality is a consequence of Proposition~\ref{prop:CS} and the definitions of the 1st Chern number and the quantum volume. The right-hand side inequality follows from an $L^2-$Cauchy-Schwarz inequality; a detailed derivation can be found in the accompanying paper~\cite{oz:mer:21:published}. Here we will just prove the last additional statement. The saturation of the inequality $\vol_{g}(\BZ^2)\leq \vol_{\widetilde{g}}(T^2_{\theta})$ is shown in~\cite{oz:mer:21:published} to be equivalent to the equality
\begin{align*}
g_{ij}(\bf{k})=g_{11}(\bf{k}) c_{ij}, \text{for every } \bf{k}\in \BZ^2,
\end{align*}
for some constants $c_{ij}$, $1\leq i,j\leq 2$. This means that $g$ and $\widetilde{g}$, defined by Eq.~\eqref{eq: twist-angle quantum metric}, differ by a Weyl rescaling -- hence, provided $P$ is an immersion so that $g$ is non-degenerate, they share the same complex structure. It then follows that
\begin{align*}
\widetilde{g}_{ij}(\theta)=\int_{\BZ^2}\frac{d^2k}{(2\pi)^2}g_{ij}(\bf{k})=\left(\int_{\BZ^2}\frac{d^2k}{(2\pi)^2}g_{11}(\bf{k})\right)c_{ij}.
\end{align*}
The result follows immediately, for
\begin{align*}
e^{2f(\bf{k})}=\frac{g_{11}(\bf{k})}{\int_{\BZ^2}\frac{d^2k}{(2\pi)^2}g_{11}(\bf{k})}.
\end{align*}
\end{proof}

\begin{theorem} For two-dimensional two-band models, there must exist a point in the Brillouin zone where $\det(g)=0$. In other words, the map $P:\BZ^2\to \Gr_{1}(\mathbb{C}^2)=\mathbb{C}P^1\cong S^2$ cannot be an immersion.
\label{th:bruno}
\end{theorem}

\begin{proof}
In this particular case, since $\Gr_1(\mathbb{C}^2)=\mathbb{C}P^1\cong S^2$ is the Bloch sphere, we get a map from a torus to the two-sphere, $\BZ^2\ni \bf{k}\mapsto n(\bf{k})\in S^2$. If this map were an immersion, then by the inverse function theorem, it would be a local diffeomorphism. As a consequence it would be a covering map.  Now Proposition~1.3.1. of Ref.~\cite{hat:02}, tells us that if we have a covering map then the induced group homomorphism in fundamental groups is injective. The fundamental group of the Brillouin torus $\BZ^2$ is the $\mathbb{Z}\oplus\mathbb{Z}$ Abelian group generated by the non-contractible loops associated with the $k_1$ and $k_2$ directions. Meanwhile, that of the sphere is the trivial group since any loop is contractible to a point. As a consequence, there can be no injective group homomorphism between these two groups. Hence, we find that there can be no immersion map from the torus to the sphere and, as a conclusion, the quantum metric must be degenerate, i.e. non-invertible, somewhere in the Brillouin zone.
\end{proof}

\bibliography{bib.bib}

\end{document}